\documentclass[12pt,onecolumn,draft]{IEEEtran} %!PN

\usepackage{amsmath,amssymb,amsfonts}
\usepackage[final]{graphicx}
\newtheorem{corollary}{Corollary}
\newtheorem{theorem}{Theorem}

\newtheorem{lemma}{Lemma}
\newtheorem{definition}{Definition}

\newtheorem{example}{Example}

\def \R{\mathbb{R}}

\def \C{\mathbb{C}}

\def \L {{\mathcal L }}

\def \M {{\mathcal M }}

\def \Q {{\mathcal Q }}

\def \E{{\textrm{E}}}

\begin{document}
%\tableofcontents
\title{\huge Unified Analysis of the Average Gaussian Error Probability for
a Class of Fading Channels
\thanks{\hspace{-1.5mm}* The author is with the Dpto. de Ingenier\'ia de Comunicaciones, Universidad de M\'alaga, Spain.
{\tt{paris@ic.uma.es}}\newline
** This work is partially supported by the Spanish Government under project TEC2007-67289/TCM
and by AT4 wireless.}}

\author{
\vspace{0mm}
\authorblockN{ Jos\'e F. Paris}}
%\\
%\vspace{0mm}
%\authorblockA{\normalsize Dept. of Ingenier\'\i a de Comunicaciones, E.T.S.I. de
%Telecomunicaci\'on,
%\\[0mm]
%Universidad de M\'alaga, M\'alaga,\\[0mm]
%E-29071, Spain}\\[0mm]
%\tt{paris@ic.uma.es} }

% El Abstract no debe superar las 50 palabras para un IEEE TCOM Transactions Letters !!!!

\maketitle
\begin{abstract}
This paper focuses on the analysis of average Gaussian error
probabilities in certain fading channels, i.e. we are interested in $\E[Q(\sqrt{p \gamma}]$
where $Q(\cdot)$ is the Gaussian $Q$-function, $p$ is a positive real number and $\gamma$ is a
nonnegative random variable.
We present a unified analysis of the average
Gaussian error probability, derive a compact
expression in terms of the Lauricella $F_D^{(n)}$ function that is applicable
to a broad class of fading channels, and
discuss the relation of this expression and expressions of this
type recently appeared in literature.
As an intermediate step in our derivations, we also obtain a compact expression for
the outage probability of the same class of fading channels.
Finally, we show how this unified analysis allows us
to obtain novel performance analytical results.
\end{abstract}

%The main purpose of this Letter is to obtain compact analytical
%expressions for the average Gaussian error probabilities
%applicable to a wide variety of fading channels.

\vspace{0mm}

\IEEEpeerreviewmaketitle

\section{Introduction}

\PARstart{F}{or} many decades, communication theorists have
analyzed the performance of single-channel and multichannel
receivers in a fading environment. Among performance measures,
the average bit error probability (BEP) is perhaps the one
that is most revealing about the nature of the system behavior.
The average BEP is defined as the average of the conditional BEP over the fading
statistics. The conditional BEP of a fading channel is often
equivalent to the BEP of an additive white Gaussian noise (AWGN) channel,
e.g. systems employing \mbox{M-AM} or \mbox{QAM} under ideal coherent detection \cite{Simon05}.
In such cases, the average BEP can be computed from the statistical expectation of Gaussian error
probabilities with respect to the fading distribution.

This paper focuses on the analysis of average Gaussian error
probabilities in fading channels, i.e. we are interested in $\E[Q(\sqrt{p \gamma}]$
where $Q(\cdot)$ is the Gaussian $Q$-function, $p$ is a positive real number and $\gamma$ is a nonnegative random variable.
We may note that the final average BEP expression is usually expressed as a weighted sum of average Gaussian error
probabilities, e.g. see \cite{chinos2003} for QAM with Gray mapping; however, for clarity these details will
be overlooked here.
The literature concerning average Gaussian error probability calculations is now quite voluminous.
The most complete account of this problem is found in \cite{Simon05}; nevertheless, a great deal
of new results have appeared after the publication of \cite{Simon05}. Some of these new closed-form results
involve the Lauricella function $F_D^{(n)}$, e.g. \cite{Lau1}-\cite{Lau3}.
The approach identified in these new contributions
is based on applying an appropriate change of variable
to the expression obtained by the well-known moment generating function (MGF) method developed in \cite{Simon05}.
Therefore, these results raise the natural question of whether the involved fading distributions
share a common property that leads to this particular mathematical form.

In this paper we identify the common property of certain fading distributions
that leads to average Gaussian probabilities expressed by the Lauricella function $F_D^{(n)}$. Such property
is related to the form of the associated MGF.
We derive unified expressions for the average Gaussian error probability and the outage probability,
which are applicable to a large class of fading distributions. Moreover, this unified analysis provides a systematic
method for obtaining new analytical results.

The remainder of this paper is organized as follows. The unified analysis is presented in \mbox{Section
\ref{analysis}}. In \mbox{Section \ref{applications}} we apply our analysis to derive
both published and novel analytical results. Finally, some conclusions are given in
\mbox{Section \ref{conclusions}.}

\section{Unified Analysis}
\label{analysis}

In this section we derive the key results of this work.
Some comments on notation are in order.
For an arbitrary function $\phi(x)$ we denote the Laplace transform as $\L[\phi(x);s]$.
As in \cite{Simon05}, we define the MGF of a nonnegative random variable $\gamma$ as
$\M_{\gamma}(s)=\E[e^{s \gamma}]=\L[f_{\gamma}(\gamma);-s]$, where
$s\in\C$ and $f_{\gamma}(\gamma)$ is the probability density function (PDF) of $\gamma$.

The following definition will be very useful for our purposes\footnotemark[1].
\addtocounter{footnote}{1}\footnotetext{This terminology is inspired by the geometric programming theory.}
\begin{definition}
[Monomial and Posynomial MGF]
\label{P MGF}
A nonnegative random variable $\gamma$ has a \mbox{\emph{posynomial}} MGF if its MGF has the form
\begin{equation}
\label{defMGF} \M_\gamma  (s) = \sum\limits_{k = 1}^K {c_k
\prod\limits_{i = 1}^{n_k } {\left( {1 - \frac{s} {{a_{k,i} }}}
\right)} ^{ - b_{k,i} } },
\end{equation}
where the involved parameters satisfy the following conditions
\begin{equation}
\label{condpos}
\left\{ \begin{gathered}
  \operatorname{Re} [a_{k,i} ] > 0,\quad i = 1,2, \ldots ,n_k \text{ and }k = 1,2, \ldots ,K, \hfill \\
  \sum\limits_{i = 1}^{n_k } {b_{k,i}  > 0} ,\quad k = 1,2, \ldots ,K, \hfill \\
  \sum\limits_{k = 1}^K {c_k  = 1}.  \hfill \\
\end{gathered}  \right.
\end{equation}
For the special case $K=1$, the MGF will be called \emph{monomial}.
\end{definition}
\vspace{2mm}
In the interest of brevity, we will say a random variable or a distribution is monomial or posynomial
if its associated MGF is
monomial or posynomial, respectively. In addition, the coefficients in Definition \ref{P MGF} will
be referred as \emph{characteristic coefficients} and
the conditions given in (\ref{condpos}) will be called \emph{compatibility
conditions}. The necessity of these conditions will be revealed along this section\footnotemark[2].
\addtocounter{footnote}{1}\footnotetext{We anticipate that the first and second compatibility conditions will be used in the
proof of Theorem \ref{teo}, while the third condition is necessary to satisfy $\M_{\gamma}(0)=1$.}
As shown in Table 1, identifying $\gamma$ with the instantaneous signal-to-noise ratio (SNR), we observe that
many common fading distribution have monomial MGF.

%Since posynomial MGFs are closed under convex combination,
%multiplication and scaling; a sum of independent posynomial random variables is posynomial, and the
%same property is straightforward to show for finite mixtures.

To grasp the generality of the results derived in this section, we may define the average Gaussian error probability
in a more fundamental form. Integral transforms theory provides a convenient framework for our purposes. In
\cite{Tulino2003} a new integral transform, called \emph{Shannon transform} was considered in relation to the
ergodic capacity analysis. Regarding the problem treated here, we may consider the following integral transform.
\begin{definition}
[Gaussian $\Q$-transform]
\label{Qtrans}
Given a nonnegative absolutely continuous random variable $\gamma$ with PDF $f_{\gamma}(\gamma)$, we define
\begin{equation}
{\Q}_\gamma  (p)\equiv\Q[f_{\gamma}(\gamma);p]=\E[Q(\sqrt{p \gamma})],
\end{equation}
where $p$ is a real nonnegative number.
\end{definition}
Some straightforward properties of the Gaussian $\Q$-transform, including an inversion formula, are collected
in the following result.
\begin{lemma}
\label{Q}
Let us consider a nonnegative absolutely continuous random variable
$\gamma$.
Then, the corresponding Gaussian $\Q$-transform is always defined in the set $\R^{+}\equiv[0,\infty)$,
and the following properties hold:\\
(i) ${\Q}_{\gamma}(p)$ is a continuous and decreasing function in $\R^{+}$ such that
${\Q}_{\gamma}(0)=\tfrac{1}{2}$ and
\mbox{$\mathop {\lim }\limits_{p \to \infty } {\Q}_\gamma  \left( p \right) = 0$.}\\
(ii) If ${\Q}_{\gamma}(p)=\Phi(p)$ then
\begin{equation*}
f_\gamma  \left( \gamma  \right) = \sqrt {2\pi } \frac{d}
{{d\gamma }}\left\{ {\sqrt \gamma  \int_{\varepsilon  - j\infty }^{\varepsilon  + j\infty } {\frac{{\Phi \left( p \right)}}
{{\sqrt p }}} e^{\frac{{p\gamma }}
{2}} dp} \right\},
\end{equation*}
for any $\varepsilon>0$.
\end{lemma}
\begin{proof}
See Appendix I.
\end{proof}
Although the Gaussian $\Q$-transform is an attractive concept, it is beyond the scope of this paper
to develop these matters further.

Under the previous theoretical framework, one can think in Gaussian $\Q$-transform pairs. The key result of this work
is conceptually summarized as follows: monomial distributions and certain expressions involving
the Lauricella $F_D^{(n)}$ function are connected by
the Gaussian $\Q$-transform. The extension of this idea to posynomial distributions is straightforward.
From an operational point of view, once the underlaying fading distribution is identified
as posynomial, obtaining the Gaussian $\Q$-transform reduces to extracting the characteristic coefficients from the MGF.
The mathematically precise statements are given below.

It is most common to find monomial distributions in practical
problems, thus, we start presenting the results for the monomial case.
\begin{theorem}
\label{teo}
Let $\gamma$ be a \emph{monomial} random variable, i.e.
\begin{equation}
\M_{\gamma}(s) = \prod\limits_{i = 1}^n {\left( {1 - \frac{s}
{{a_i }}} \right)} ^{ - b_i },
\end{equation}
where $\{a_i\}_{i=1}^n$  and $\{b_i\}_{i=1}^n$ satisfy the compatibility conditions given in (\ref{condpos}). Then\\
(i) The cumulative distribution function (CDF) of $\gamma$ is expressed as
\begin{equation}
\label{Key1}
F_\gamma  \left( \gamma  \right) = \left\{ {\frac{{\prod\limits_{i = 1}^n {\left( {a_i } \right)^{b_i } } }}
{{\Gamma \left( {1 + \sum\limits_{i = 1}^n {b_i } } \right)}}} \right\}\gamma ^{\sum\limits_{i = 1}^n {b_i } } \Phi _2^{(n)}
\left( {b_1 , \ldots ,b_n ;1 + \sum\limits_{i = 1}^n {b_i } ; - a_1 \gamma , \ldots , - a_n \gamma } \right),
\end{equation}
where $\Gamma$ is the gamma function and $\Phi _2^{(n)}$ is the confluent Lauricella function defined in \cite{Srivastava1985}-\cite{Erdelyi1954}.\\
(ii) The Gaussian $\Q$-transform of $\gamma$ is given by
\begin{equation}
\label{Key2}
\begin{gathered}
  {\Q}_\gamma  (p) = \frac{1}
{{2\sqrt \pi  }}\left\{ {\prod\limits_{i = 1}^n {\left( {a_i } \right)^{b_i } } } \right\}\left\{ {\frac{{\Gamma \left( {\frac{1}
{2} + \sum\limits_{i = 1}^n {b_i } } \right)}}
{{\Gamma \left( {1 + \sum\limits_{i = 1}^n {b_i } } \right)}}} \right\}\left( {\frac{2}
{p}} \right)^{\sum\limits_{i = 1}^n {b_i } }  \hfill \\
  \quad  \times F_D^{(n)} \left( {\frac{1}
{2} + \sum\limits_{i = 1}^n {b_i } ,b_1 , \ldots ,b_n ;1 + \sum\limits_{i = 1}^n {b_i } ; - \frac{{2a_1 }}
{p}, \ldots , - \frac{{2a_n }}
{p}} \right), \hfill \\
\end{gathered}
\end{equation}
where $F_D^{(n)}$ is the Lauricella function defined in \cite{Srivastava1985}-\cite{Erdelyi1954}.
\end{theorem}
\begin{proof}
See Appendix II.
\end{proof}
Extending this last result to posynomial random variables is straightforward.
\begin{corollary}
\label{cor1}
If $\gamma$ is a \emph{posynomial} random variable such that its MGF is given by (\ref{defMGF}) then\\
(i) The CDF of $\gamma$ is expressed as
\begin{equation}
\label{Key12}
\begin{gathered}
  F_\gamma  \left( \gamma  \right) = \sum\limits_{k = 1}^K {c_k } \left\{ {\frac{{\prod\limits_{i = 1}^n {\left( {a_{k,i} } \right)^{b_{k,i} } } }}
{{\Gamma \left( {1 + \sum\limits_{i = 1}^n {b_{k,i} } } \right)}}} \right\}\gamma ^{\sum\limits_{i = 1}^n {b_{k,i} } }  \hfill \\
  \quad  \times \Phi _2^{(n_k )} \left( {b_{k,1} , \ldots ,b_{k,n_k } ;1 +
  \sum\limits_{i = 1}^{n_k } {b_{k,i} } ; - a_{k,1} \gamma , \ldots , - a_{k,n_k } \gamma } \right). \hfill \\
\end{gathered}
\end{equation}
(ii) The Gaussian $\Q$-transform of $\gamma$ is given by
\begin{equation}
\label{Key22}
\begin{gathered}
  Q_\gamma  (p) = \frac{1}
{{2\sqrt \pi  }}\sum\limits_{k = 1}^K {c_k } \left\{ {\prod\limits_{i = 1}^{n_k } {\left( {a_{k,i} } \right)^{b_{k,i} } } } \right\}\left\{ {\frac{{\Gamma \left( {\frac{1}
{2} + \sum\limits_{i = 1}^{n_k } {b_{k,i} } } \right)}}
{{\Gamma \left( {1 + \sum\limits_{i = 1}^{n_k } {b_{k,i} } } \right)}}} \right\}\left( {\frac{2}
{p}} \right)^{\sum\limits_{i = 1}^{n_k } {b_{k,i} } }  \hfill \\
  \quad  \times F_D^{(n_k )} \left( {\frac{1}
{2} + \sum\limits_{i = 1}^{n_k } {b_{k,i} } ,b_{k,1} , \ldots ,b_{k,n_k } ;1 + \sum\limits_{i = 1}^{n_k } {b_{k,i} } ; - \frac{{2a_{k,1} }}
{p}, \ldots , - \frac{{2a_{k,n_k} }}
{p}} \right). \hfill \\
\end{gathered}
\end{equation}
\end{corollary}
\begin{proof}
Repeat the steps of the proof of Theorem \ref{teo}.
\end{proof}

Theorem \ref{teo} and Corollary \ref{cor1} have important theoretical and practical consequences.
Next result provides an integral representation for the Gaussian $\Q$-transform of posynomial
random variables that is convenient for numerical evaluation.

\begin{corollary}
\label{cor2}
If $\gamma$ is a \emph{posynomial} random variable such that its MGF is given by (\ref{defMGF}) then\\
\begin{equation}
\label{cor333}
\begin{gathered}
  \Q_\gamma  (p) = \frac{1}
{{2\pi }}\sum\limits_{k = 1}^K {c_k } \left\{ {\prod\limits_{i = 1}^{n_k } {\left( {a_{k,i} } \right)^{b_{k,i} } } } \right\}\left( {\frac{2}
{p}} \right)^{\sum\limits_{i = 1}^{n_k } {b_{k,i} } }  \hfill \\
  \quad  \times \int_0^1 {u^{ - \frac{1}
{2} + \sum\limits_{i = 1}^{n_k } {b_{k,i}  } } \left( {1 - u} \right)^{ - \frac{1}
{2}} \left( {1 + \frac{{2a_{k,1} }}
{p}u} \right)^{ - b_{k,1} }  \ldots \left( {1 + \frac{{2a_{k,n} }}
{p}u} \right)^{ - b_{k,n_k } } du}.  \hfill \\[4mm]
\end{gathered}
\end{equation}
\end{corollary}
\begin{proof}
Check that the compatibility conditions allow us to use the Euler-type integral representation given in
\cite[p. 283, eq. 34]{Srivastava1985} for the Lauricella functions in (\ref{Key22}).
\end{proof}

From this last Corollary it is straightforward to obtain an
asymptotic approximation for posynomial random variables.

\begin{corollary}
\label{cor3}
Let $\gamma$ be a \emph{posynomial} random variable such that its MGF is given by (\ref{defMGF}).
Then
\begin{equation}
\label{asym}
\Q_\gamma  (p) \sim \frac{1}
{{2\sqrt \pi  }}\sum\limits_{k = 1}^K {c_k } \left\{ {\prod\limits_{i = 1}^{n_k } {\left( {2a_{k,i} } \right)^{b_{k,i} } } } \right\}\frac{{\Gamma \left( {\frac{1}
{2} + \sum\limits_{i = 1}^{n_k } {b_{k,i}  } } \right)}}
{{\Gamma \left( {1 + \sum\limits_{i = 1}^{n_k } {b_{k,i} } } \right)}}\frac{1}
{{p^{\sum\limits_{i = 1}^{n_k } {b_{k,i} } } }},
\end{equation}
for large values of $p$.
\end{corollary}
\begin{proof}
Substitute the integrand in (\ref{cor333}) by its asymptotic approximation and use the basic properties
of the beta function \cite[p. 898-899]{Gradsh00}.
\end{proof}
As expected, from (\ref{asym}) it is inferred that the underlaying diversity order of $\Q_\gamma  (p)$ is
\begin{equation}
\mathop {\lim }\limits_{p \to \infty } \; - \frac{{\log \left( {\Q_\gamma  (p)} \right)}}
{{\log \left( p \right)}} = \mathop {\min }\limits_k \left\{ {\sum\limits_{i = 1}^{n_k } {b_{k,i} } } \right\},
\end{equation}
assuming $\gamma$ is a posynomial random variable.

\section{Applications}
\label{applications}

The mathematical tools developed in previous section provide us a unified analytical framework for many existing
results in literature; in addition, they systematically allow us to obtain new analytical results.

To calculate the average BEP, its asymptotic approximation and the outage probability
we can follow four steps:
\begin{enumerate}
\item Check if the underlaying distribution is posynomial. In such case extract the characteristic coefficients.
\item Use the Theorem and Corollaries of previous section to obtain analytical expressions for the Gaussian $\Q$-Transform, its asymptotic
approximation and the outage probability.
\item If possible, reduce the Lauricella functions $F_D^{(n)}$ and $\Phi _2^{(n)}$ to simpler functions.
\item Use known formulas to express the average BEP in terms of the Gaussian $\Q$-transform,
e.g. see \cite{chinos2003} for QAM with Gray mapping.
\end{enumerate}

Next, we provide some examples to illustrate how this approach allows us
to derive both published and novel results. For brevity we will mainly focus on the first step described above.

\begin{example}\emph{Nakagami-$m$ fading.}
This a well-known example \cite{Simon05}\cite{Eng95}.
From Table I, we can extract the characteristic parameters for the Nakagami-$m$ distribution; specifically,
$a_1=m/\bar \gamma$ and $b_1=m$. From (\ref{Key1}) we obtain
\begin{equation}
\label{ejemplo-1-1}
F_\gamma  \left( \gamma  \right) = \left\{ {\frac{{\left( {\frac{m}
{{\bar \gamma }}} \right)^m }}
{{\Gamma \left( {1 + m} \right)}}} \right\}\gamma ^m {}_1F_1 \left( {m,1 + m; - \frac{m}
{{\bar \gamma }}\gamma } \right).
\end{equation}
Then, after considering \cite[eq. 8.351-2]{Gradsh00}, we derive the well-known formula
\begin{equation}
F_\gamma  \left( \gamma  \right) =1 - \frac{{\Gamma \left( {m,\frac{m}
{{\bar \gamma }}\gamma } \right)}}
{{\Gamma \left( m \right)}}.
\end{equation}
Applying (\ref{Key2}) and considering \cite[eq. 9.131-1]{Gradsh00} yields
\begin{equation}
\label{ejemplo-1-2}
{\Q}_\gamma  (p) = \frac{1}
{{2\sqrt \pi  }}\left\{ {\frac{{\Gamma \left( {\frac{1}
{2} + m} \right)}}
{{\Gamma \left( {1 + m} \right)}}} \right\}\left( {\frac{1}
{c(p)}} \right)^m \left( {1 + \frac{1}
{c(p)}} \right)^{1/2 - m} {}_2F_1 \left( {\frac{1}
{2},1;m + 1; - \frac{1}
{c(p)}} \right),
\end{equation}
where $c(p) \doteq {{p \bar \gamma}}/{{(2m)}}$.
This expression is very similar to that given \cite[eq. 5.A.2]{Simon05}\cite[eq. A.8]{Eng95}, which was derived
by a different approach. Interestingly,
it has been checked numerically that both are equivalent; however, the author has not been able to find the transformation
between Gauss hypergeometric functions ${}_2F_1$ connecting (\ref{ejemplo-1-2}) and \cite[eq. 5.A.2]{Simon05}.
\hspace{36mm}\QEDopen
\end{example}

\begin{example}\emph{Maximal ratio combining (MRC) with independent but nonidentically distributed
branches.}
A formula for the Gaussian $\Q$-transform over Hoyt fading channels is derived in \cite{Lau2} in terms of the
Lauricella $F_D^{(n)}$ function. Now we show how to perform a unified analysis for
MRC including Hoyt distributed branches.
Let us consider receive MRC with $L$ independent but nonidentical distributed branches.
We assume that $L_1$ branches
exhibit Nakagami-$q$ (Hoyt) fading and the remainder of the $L-L_1$
branches are better modelled by the Nakagami-$m$ fading model.
The Nakagami-$m$ fading model is used for both line-of-sight
(LOS) and NLOS scenarios, while the Nakagami-$q$ distribution
is an alternative to model NLOS channels.
We assume a general scenario where each branch has arbitrary parameters, i.e. $\{q_j,\bar\gamma_j\}_{j=1}^{L_1}$
for the first $L_1$ branches and $\{m_j,\bar\gamma_j\}_{j=L_1 +1}^{L}$ for the remainder of the branches.
Analyzing such scenario is straightforward with the results derived in this work,
after observing that the associated MGF is monomial
\begin{equation}
\M_\gamma  (s) = \prod\limits_{i = 1}^{L_1 } {\left( {1 - \frac{{2q_i^2 }}
{{1 + q_i^2 }}\bar \gamma _i s} \right)^{ - 1/2} \prod\limits_{i = L_1  + 1}^{2L_1 } {\left( {1 - \frac{2}
{{1 + q_i^2 }}\bar \gamma _i s} \right)^{ - 1/2} } } \prod\limits_{i = 2L_1  + 1}^{L_1  + L } {\left( {1 - \frac{{\bar \gamma _i }}
{{m_i }}s} \right)^{ - m_i } }.
\end{equation}
This scenario can be further generalized to consider large-scale time variations. Let us assume that the number
of branches with Nakagami-$q$ fading $L_1$ is also random. Then, if $\Pr[l]$ represents the probability of having
$l$ branches with Nakagami-$q$ fading, the associated MGF is now posynomial
\begin{equation}
\begin{gathered}
  \M_\gamma  (s) = \sum\limits_{l = 1}^L {\Pr [l]\prod\limits_{i = 1}^l {\left( {1 - \frac{{2q_{l,i}^2 }}
{{1 + q_{l,i}^2 }}\bar \gamma _{l,i} s} \right)^{ - 1/2} } }  \hfill \\
  \quad  \times \prod\limits_{i = l + 1}^{2l} {\left( {1 - \frac{2}
{{1 + q_{l,i}^2 }}\bar \gamma _{l,i} s} \right)^{ - 1/2} } \prod\limits_{i = 2l + 1}^{l + L} {\left( {1 - \frac{{\bar \gamma _{l,i} }}
{{m_{l,i} }}s} \right)^{ - m_{l,i} } }.  \hfill \\
\end{gathered}
\end{equation}
A system employing MRC with independent $\eta$-$\mu$ distributed
branches has been recently analyzed in \cite{Lau3}. This system
generalizes the Nakagami-$q$/Nakagami-$m$ scenario described above
and it requires to estimate the set of $3L$ parameters
$\{\eta_j,\mu_j,\bar\gamma_j\}_{j=1}^{L}$. Since the
\mbox{Nakagami-$q$/Nakagami-$m$} scenario requires $2L$ channel
parameters, it is a reasonable alternative for certain
applications. Interestingly, if we apply Theorem \ref{teo} to the
monomial MGF considered in \cite[eq. 2]{Lau3}, we can extend the
performance analysis in \cite{Lau3} with a closed-form expression
for the outage probability and an asymptotic approximation for the
average BEP.\hspace{67mm}\QEDopen
\end{example}

\begin{example}\emph{Orthogonal space-time block codes (OSTBC) in spatially correlated multiple-input
multiple-output (MIMO) channels.} In \cite{Wong2006} the
underlying MGF was derived for different systems using OSTBC in
shadowed Rician MIMO fading channels. Such MGFs are posynomial,
thus, all results derived in previous section are applicable,
leading to novel analytical results. To illustrate this idea, we
consider the particular MGF for correlated LOS component and
spatially white scattered component \cite[eq. 23]{Wong2006}
\begin{equation}
\label{ejemplo3}
\begin{gathered}
  \M_\gamma  (s) = \frac{1}
{{\left( {1 + as} \right)^{n_r n_t } }}\prod\limits_{i = 1}^{n_t n_r } {\left( {1 + \frac{{bs}}
{{1 + as}}} \right)^{ - m} }  \hfill \\
  \quad  = \left( {1 + as} \right)^{ - n_r n_t } \prod\limits_{i = 1}^{n_t n_r }
  {\left( {1 + (a + b)s} \right)^{ - m} } \prod\limits_{i = 1}^{n_t n_r } {\left( {1 + as} \right)^m },  \hfill \\
\end{gathered}
\end{equation}
where $n_t$,$n_r$,$a$,$b$ and $m$ are channel and system parameters defined in \cite{Wong2006}. After rewriting
the MGF given in \cite[eq. 23]{Wong2006}, we clearly observe in (\ref{ejemplo3}) that it is a monomial MGF.
\hspace{17mm}\QEDopen
\end{example}

\section{Conclusions}
\label{conclusions}
A solid analytical framework for the computation of average Gaussian error probabilities has been derived
in this paper.
This quantity allows us to obtain the average BEP in a variety of wireless communications systems.
It has been shown that if the MGF of the underlaying fading distribution has certain posynomial structure, the
average Gaussian error probabilities and the outage probability can be expressed in terms of
Lauricella functions. Finally, this analytical framework has been used to derive known and novel results in
a simple and systematic way.

\appendices

\section{Proof of Lemma \ref{Q}}

The existence of the Gaussian $\Q$-transform is clear after taking into account
that $Q(\sqrt{p\gamma})\leq \tfrac{1}{2}$ for every $p>0$.

(i) Continuity follows form the fact that the kernel $Q(\sqrt{p\gamma})$ is continuous and bounded in $\R^{+}$.
If $p_1<p_2$ then $Q(\sqrt{p_1\gamma})>Q(\sqrt{p_2\gamma})$ for all $\gamma$ in $\R^{+}$, thus,
$\Q_{\gamma}(p_1)>\Q_{\gamma}(p_2)$. The property
$\mathop {\lim }\limits_{p \to \infty } {\Q}_\gamma  \left( p \right) = 0$ follows after
applying the monotone convergence theorem.

(ii) After integrating by parts $Q_{\gamma}(p)$ we can write\\
\begin{equation}
\label{appex1}
\Phi (p) = \frac{{\sqrt p }}
{{2\sqrt {2\pi } }}\int_0^\infty  {e^{ - \frac{{pt}}
{2}} F_\gamma  (t)t^{ - 1/2} dt}.
\end{equation}
The integral in (\ref{appex1}) is recognized as a Laplace transform converging for $\Re[p]>0$. Thus, after applying the inverse Laplace
transform to this equality, the desired result is obtained.

\section{Proof of Theorem \ref{teo}}

(i) Since $\L[f_\gamma  \left( t \right);s] = \prod\limits_{i = 1}^n {\left( {1 + \frac{s}
{{a_i }}} \right)} ^{ - b_i }$, we can write
\begin{equation}
\label{appex2}
\L[F_\gamma  \left( t \right);s] = \frac{1}
{s}\L[f_\gamma  \left( t \right);s] =
\left\{ {\frac{{\prod\limits_{i = 1}^n {\left( {a_i } \right)^{b_i } } }}
{{\Gamma \left( {\sum\limits_{i = 1}^n {b_i } } \right)}}} \right\}\left\{ {\frac{{\Gamma \left( {\sum\limits_{i = 1}^n {b_i } } \right)}}
{{s^{\sum\limits_{i = 1}^n {b_i } } }}} \right\}\prod\limits_{i = 1}^n {\left( {1 - \frac{{( - a_i )}}
{s}} \right)} ^{ - b_i },
\end{equation}
for $\Re[s]>0$.
Taking into account
the first and second compatibility conditions given in (\ref{condpos}) and
identifying \cite[p. 222, eq. 5]{Erdelyi1954} from (\ref{appex2}), the expression (\ref{Key1}) is obtained.

(ii) Introducing (\ref{Key1}) in (\ref{appex1}) yields
\begin{equation}
\label{appex3}
{\Q}_\gamma  (p) = \frac{{\sqrt p }}
{{2\sqrt {2\pi } }}\left\{ {\frac{{\prod\limits_{i = 1}^n {\left( {a_i } \right)^{b_i } } }}
{{\Gamma \left( {\sum\limits_{i = 1}^n {b_i } } \right)}}} \right\}\int_0^\infty  {e^{ - \frac{p}
{2}t} t^{ - 1 - 1/2 + \sum\limits_{i = 1}^n {b_i } } \Phi _2^{(n)}
\left( {b_1 , \ldots ,b_n ,\sum\limits_{i = 1}^n {b_i } ; - a_1 t, \ldots , - a_n t} \right)dt}.
\end{equation}
Taking into account
the first and second compatibility conditions given in (\ref{condpos}) and
identifying \cite[p. 286, eq. 43]{Srivastava1985} from (\ref{appex3}),
the expression (\ref{Key2}) is obtained.

\newpage

\begin{table}[t!]
\caption{Some Monomial MGFs.}
\begin{center}
\small{\begin{tabular}{c | c | c | c | c | c | c | c | c }

\hline\hline
\multicolumn{3}{ c }{Fading Distribution} & \multicolumn{3}{ c }{MGF $\M_{\gamma}(s)$}
& \multicolumn{3}{ c }{References}\\

\hline\hline
\multicolumn{3}{c |}{Rayleigh}
& \multicolumn{3}{| c |}{
$
\begin{gathered}
\hfill \\[-5mm]
  \left( {1 - \bar \gamma s} \right)^{ - 1} ;\hfill\\
  \quad 0 \leqslant \bar \gamma .
\end{gathered}
$
}
& \multicolumn{3}{| c }{\cite{Rayleigh},\cite{Simon05}}\\
\hline
\multicolumn{3}{c |}{Nakagami-$q$ (Hoyt)}
& \multicolumn{3}{| c |}{
$
\begin{gathered}
\hfill \\[-5mm]
  \left( {1 - \frac{{2q^2 }}
{{1 + q^2 }}\bar \gamma s} \right)^{ - 1/2} \left( {1 - \frac{2}
{{1 + q^2 }}\bar \gamma s} \right)^{ - 1/2};\hfill \\
 0 \leqslant \bar \gamma, 0 < q \leqslant 1.
\end{gathered}
$
}
& \multicolumn{3}{| c }{\cite{Hoyt1947},\cite{Simon05}}\\
\hline
\multicolumn{3}{c |}{Nakagami-$m$}
& \multicolumn{3}{| c |}{
$
\begin{gathered}
\hfill \\[-6mm]
\left( {1 - \frac{{\bar \gamma }}
{m}s} \right)^{ - m} ;\hfill\\
\quad\tfrac{1}
{2} \leqslant m.
\end{gathered}
$
}
& \multicolumn{3}{| c }{\cite{Nakagami1960},\cite{Simon05}}\\
\hline
\multicolumn{3}{c |}{Rician Shadowed$^{\star}$}
& \multicolumn{3}{| c |}{
$
\begin{gathered}
\hfill \\[-5mm]
  \left( {1 - \frac{{\bar \gamma }}
{{1 + K}}s} \right)^{m - 1} \left( {1 - \left( {1 + \frac{K}
{m}} \right)\frac{{\bar \gamma }}
{{1 + K}}s} \right)^{ - m};  \hfill \\
  0 \leqslant \bar \gamma , 0 \leqslant K, 0 < m.
\end{gathered}
$
}
& \multicolumn{3}{| c }{\cite{Abdi03},\cite{Simon05}}\\
\hline
\multicolumn{3}{c |}{$\eta$-$\mu$ Physical Model}
& \multicolumn{3}{| c |}{
$
\begin{gathered}
\hfill \\[-5mm]
  \left( {1 - s\frac{{\bar \gamma }}
{{n\left( {h + H} \right)}}} \right)^{ - n/2} \left( {1 - s\frac{{\bar \gamma }}
{{n\left( {h - H} \right)}}} \right)^{ - n/2}  \hfill \\
  0 \leqslant \bar \gamma ,n = 1,2, \ldots  \hfill \\
  \text{Format 1} \Rightarrow h = \frac{{2 + \eta ^{ - 1}  + \eta }}
{4},H = \frac{{\eta ^{ - 1}  - \eta }}
{4},0 < \eta  < \infty  \hfill \\
  \text{Format 2} \Rightarrow h = \frac{1}
{{1 - \eta ^2 }},H = \frac{\eta }
{{1 - \eta ^2 }}, - 1 < \eta  < 1 \hfill\\
\end{gathered}
$
}
& \multicolumn{3}{| c }{\cite{Yacoub07}}\\
\hline\hline

\end{tabular}}
\small{\begin{tabular}{l}
{{\normalsize$^{\star}$} Note that $K\doteq \frac{\Omega}{2b_0}$, using the same notation as in \cite{Abdi03},\cite{Simon05}.}\\
\end{tabular}}
\end{center}
\end{table}

%\begin{figure}[t!]
%\centering\includegraphics[width=9cm]{SWC.eps}
%\begin{center}
%\caption{ \footnotesize
%Outage probability versus normalized average SINR. In this plot three interferers are considered with $W_1=1/4$
%and $W_2=W_3=1/8$, while the background noise power is $\sigma^2=1/10$.
%}
%\label{Grafica_asintotica}
%\end{center}
%\end{figure}

\end{document}